\documentclass[a4paper,UKenglish]{article}
 
\usepackage{tikz}
\usetikzlibrary{calc}
\usetikzlibrary{decorations.markings}
\usepackage{subfig}
\usepackage{hyperref}
\usepackage{booktabs}
\usepackage{enumitem}
\usepackage{color}
\usepackage{amssymb,amsmath,amsthm}

\usepackage{fullpage}

\newtheorem{theorem}{Theorem}
\newtheorem{lemma}[theorem]{Lemma}
\newtheorem{proposition}[theorem]{Proposition}
\newtheorem{corollary}[theorem]{Corollary}
\newtheorem{conjecture}{Conjecture}

\title{Large Independent Sets in Triangle-Free Planar Graphs\thanks{A preliminary version appeared in the Proceedings of the European Symposium on Algorithms~\cite{DvorakMnich2014}.}}

\author{Zden\v{e}k Dvo\v{r}\'{a}k
        \thanks{Computer Science Institute, Charles University, Prague, Czech Republic, \texttt{rakdver@iuuk.mff.cuni.cz}.
        Supported by the project LL1201 (Complex Structures: Regularities in Combinatorics and Discrete Mathematics) of the Ministry of Education of Czech Republic,
	and by project GA14-19503S (Graph coloring and structure) of Czech Science Foundation.}
   \and Matthias Mnich\thanks{Cluster of Excellence MMCI, Campus E1-7, 66123 Saarbr{\"u}cken, \texttt{mmnich@mmci.uni-saarland.de}}
}

\begin{document}

\def\GG {{\cal G}}
\def\LL {{\cal L}}
\def\PP {{\cal P}}
\def\SS {{\cal S}}

\maketitle

\begin{abstract}
   Every triangle-free planar graph on $n$ vertices has an independent set of size at least $(n+1)/3$, and this lower bound is tight.
   We give an algorithm that, given a triangle-free planar graph $G$ on $n$ vertices and an integer $k\ge 0$, decides whether $G$ has an independent set of size at least $(n+k)/3$, in time $2^{O(\sqrt{k})}n$.
   Thus, the problem is fixed-parameter tractable when parameterized by $k$.
   Furthermore, as a corollary of the result used to prove the correctness of the algorithm, we show that there exists $\varepsilon>0$ such that every planar graph of girth at least five on $n$ vertices has an independent set of size at least $n/(3-\varepsilon)$.
\end{abstract}

\medskip


\smallskip

\noindent
\textbf{AMS subject classifcations:} 68Q25, 68W05, 68R10.

\pagestyle{plain}

\section{Introduction}
\label{sec:introduction}
Every planar graph is 4-colorable by the deep Four-Colour-Theorem, whose proof was first announced by Appel and Haken in 1976~\cite{AppelHaken1976} and later simplified by Robertson, Sanders, Seymour and Thomas~\cite{RobertsonEtAl1997}.  As a corollary, every planar graph on $n$ vertices has an independent
set of size at least $n/4$.
The proof by Robertson et al.~\cite{RobertsonEtAl1997} comes with a quadratic-time algorithm to find a valid coloring of~$G$ with~4 colors, which can be used to find such an independent set.
Yet, determining the size of a maximum independent set is $\mathsf{NP}$-complete in planar graphs (even if they are triangle-free~\cite{Madhavan1984}).
This motivates a search for an efficient algorithm that decides, for a fixed parameter $k\ge 0$ and an input $n$-vertex planar graph~$G$, whether $G$ has an independent set of size at least $(n + k)/4$.

The problem---which we call {\sc Planar Independent Set Above Tight Lower Bound}, or {\sc Planar Independent Set-ATLB} for
short---has received a lot of attention, although there has been essentially no progress.
In fact, the question of whether {\sc Planar Independent Set-ATLB} is fixed-parameter tractable has been raised several times: first by Niedermeier~\cite{Niedermeier2006}, later by Bodlaender et al.~\cite{BodlaenderEtAl2008}, Mahajan et al.~\cite{MahajanEtAl2009}, by Sikdar~\cite{Sikdar2010}, by Mnich~\cite{Mnich2010}, and by Crowston et al.~\cite{CrowstonEtAl2011}.
Then the problem was raised again in June 2012 as a ``tough customer'', at WorKer 2012~\cite{FellowsEtAl2012}.
We remark that until now, there is not even a polynomial-time algorithm known for the case of $k = 1$, and finding such an algorithm has been an open problem for more than 30 years.
Yet, the existence of such an algorithm for $k = 1$ is certainly a necessary condition for the fixed-parameter tractability of {\sc Planar Independent Set-ATLB}.
The lower bound of $n/4$ on the size of a maximum independent set is tight for an infinite family of planar graphs:
for example, take a set of copies of $K_4$ or $C_8^2$ (the complement of the cube) or the icosahedron, and connect these copies arbitrarily in a planar way.

In this paper, we resolve the analogous problem for triangle-free planar graphs.
By a theorem of Gr{\"o}tzsch~\cite{Grotzsch1959}, every triangle-free planar graph is 3-colorable, and thus admits an independent set that contains at least one-third of its vertices.  Such a coloring, and thus also an independent set, can be found in linear time~\cite{DvorakEtAl2011}.
Later, Steinberg and Tovey~\cite{SteinbergTovey1993} showed that every triangle-free planar graph contains a \emph{non-uniform} 3-coloring, where one color class is guaranteed to contain one vertex more than the other two color classes.
Thus, any $n$-vertex triangle-free planar graph contains an independent set of size at least $(n+1)/3$, when $n\geq 3$.
On the other hand, Jones~\cite{Jones1985} found triangle-free planar graphs on $n$ vertices (for any $n\ge 2$ such that $n\equiv 2\pmod 3$) with maximum independent sets of size exactly $(n+1)/3$; see Figure~\ref{fig:tighttrianglefreegraphs}.
\tikzset{vertex/.style={minimum size=2pt,circle,fill=black,draw, inner sep=0pt},
         decoration={markings,mark=at position .5 with {\arrow[black,thick]{stealth}}}}
\begin{figure}[hbt]
  \centering
  \begin{tikzpicture}[line width=1pt, scale=0.8, rotate=90]
  \path[use as bounding box] (-1.5,1) rectangle (2.5,-7);
    \node (1a) at (0,0)[vertex]{};
    \node (1b) at (1,0)[vertex]{};
    \node (1c) at (2,0)[vertex]{};
    \node (1d) at (1,-1)[vertex]{};
    \node (1e) at (0,-1)[vertex]{};
    
    \node (2c) at (2,-1)[vertex]{};
    \node (2d) at (1,-2)[vertex]{};
    \node (2e) at (0,-2)[vertex]{};
    
    \draw(1a)--(1b);
    \draw(1b)--(1c);
    \draw(1c)--(1d);
    \draw(1d)--(1e);
    \draw(1e)--(1a);
    
    \draw(1d)--(2c);
    \draw(2c)--(2d);
    \draw(2d)--(2e);
    \draw(2e)--(1a);
    
    \node at (1,-2.5){$\vdots$};

    \node (3a) at (0,-3)[vertex]{};
    \node (3b) at (1,-3)[vertex]{};
    \node (3c) at (2,-3)[vertex]{};
    \node (3d) at (1,-4)[vertex]{};
    \node (3e) at (0,-4)[vertex]{};
    
    \draw(3a)--(3b);
    \draw(3b)--(3c);
    \draw(3c)--(3d);
    \draw(3d)--(3e);
    \draw(3e)--(3a);

    \draw (90:0cm) .. controls (220: 6cm)and (290: 11cm) ..(2cm, -1cm);    
    \draw (90:-2.5cm) .. controls (250: 3cm)and (270: 8cm) ..(2cm, -3cm);
  \end{tikzpicture}
\caption{Triangle-free planar graphs on $n$ vertices with maximum independent sets of size exactly~$(n+1)/3$.}
\label{fig:tighttrianglefreegraphs}
\end{figure}
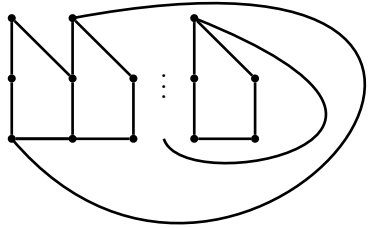
This motivates a search for an efficient algorithm that decides, for a given $n$-vertex triangle-free planar graph~$G$ and integer $k\ge 0$, whether $G$ has an independent set of size at least $(n + k)/3$.
In particular it was open whether for $k = 2$ there is a polynomial-time algorithm.
Notice that it is non-trivial even to solve this problem in time $n^{O(k)}$, as a brute-force approach does not suffice.

As our main result, we show that the problem is fixed-parameter tractable when parameterized by $k$.
\begin{theorem}
\label{thm:mainfpt-trianglefree}
  There is an algorithm that, given any $n$-vertex triangle-free planar graph $G$ and any integer $k\ge 0$, in time $2^{O(\sqrt{k})}n$ decides whether $G$ has an independent set of size at least $(n+k)/3$.
\end{theorem}

Though many different techniques have been devised for solving optimization problems parameterized above lower bounds in fixed-parameter time, none of these techniques is applicable to our problem.
Instead, our algorithm seems to be the first fixed-parameter algorithm for problems parameterized above guarantee that is based on the tree-width of a graph.
The algorithm is an easy corollary of the following result.
\begin{theorem}
\label{thm:indsize}
  There is a constant $c>0$ such that every planar triangle-free graph on $n$ vertices with tree-width $\geq t$ has an independent set of size~$\geq \frac{n+ct^2}{3}$.
\end{theorem}

According to Theorem~\ref{thm:indsize}, the algorithm of Theorem~\ref{thm:mainfpt-trianglefree} is extremely simple.  First, we test whether
the treewidth of $G$ is $O\bigl(\sqrt{k}\bigr)$, using the linear-time constant factor approximation algorithm of Bodlaender et al.~\cite{BodlaenderEtAl2013}.
If that is the case, we find the largest independent set in~$G$ by dynamic programming in time $2^{O\left(\sqrt{k}\right)}n$,
see e.g. the book of Niedermeier~\cite{Niedermeier2006}.
Otherwise, we answer ``yes''.  If we want to report the independent set of size $(n+k)/3$ when it exists, the algorithm becomes more involved,
requiring an inspection of the proof of Theorem~\ref{thm:indsize}.  We give the details in Section~\ref{sec-algo}.

As a by-product of the proof of Theorem~\ref{thm:indsize}, we also obtain the following result which is of independent interest.
\begin{theorem}
\label{thm:large-ind}
  There exists a constant $\varepsilon>0$ such that every triangle-free planar graph on~$n$ vertices and without separating $4$-cycles has an independent set of size at least $\frac{n}{3-\varepsilon}$.
\end{theorem}

Let us note an obvious consequence.

\begin{corollary}
\label{cor:large-ind}
  There exists a constant $\varepsilon>0$ such that every planar graph of girth at least $5$ on~$n$ vertices has an independent set of size at least $\frac{n}{3-\varepsilon}$.
\end{corollary}

A well-known graph parameter giving a lower bound for the independence number of a graph $G$ is the \emph{fractional chromatic number} $\chi_f(G)$, see~\cite{ScheinermanUllman2011} for a definition and other
properties.  Here, let us only recall that $\alpha(G)\ge V(G)/\chi_f(G)$.
The fractional chromatic number of planar graphs of girth at least~$8$ is at most $5/2$ (Dvo\v{r}\'ak et al.~\cite{dvopeter}), implying that for girth 8,
we can set $\varepsilon=1/2$. 
Not much is known about fractional chromatic number of graphs of girth between~$5$ and $7$, but Theorem~\ref{thm:large-ind} and Corollary~\ref{cor:large-ind} make the following conjectures plausible.

\begin{conjecture}
\label{conj:frac1}
  There exists a constant $\varepsilon>0$ such that every triangle-free planar graph without separating $4$-cycles has fractional chromatic number at most $3-\varepsilon$.
\end{conjecture}

\begin{conjecture}
\label{conj:frac2}
  There exists a constant $\varepsilon>0$ such that every planar graph of girth $\geq 5$ has fractional chromatic number at most $3-\varepsilon$.
\end{conjecture}

Note that if Conjecture~\ref{conj:frac2} holds, then $\varepsilon\le 1/4$ by a construction of Pirnazar and Ullman~\cite{PirnazarUllman2002}.
Furthermore, the girth assumption cannot be reduced to $4$ (or replaced by assuming odd girth at least~$5$) because of the construction in Fig.~\ref{fig:tighttrianglefreegraphs}.
Although Conjecture~\ref{conj:frac1} appears to be stronger than Conjecture~\ref{conj:frac2}, $4$-faces are usually easy to deal with in coloring arguments (by identifying a pair of
opposite vertices), and thus we believe the two conjectures to be equivalent.

Let us also remark that the assumption of Corollary~\ref{cor:large-ind} that $G$ is planar can be relaxed, since every graph on $n$ vertices embeddable
in a surface of genus $g$ can be planarized by removing $O(\sqrt{gn})$ vertices~\cite{Djidjev1984}.

\begin{corollary}
  There exist constants $\varepsilon,c>0$ such that every graph of girth $\geq 5$ and genus $\leq g$ on~$n$ vertices has an independent set of size $\geq \frac{n}{3-\varepsilon}(1-c\sqrt{g/n})$.
\end{corollary}

\subsection{Organization and proof outline.}
In Section~\ref{sec-tw}, we recall some basic facts about tree-width, especially as relates to planar graphs.
In Section~\ref{sec:classesofgraphswithboundedexpansion}, we review some results on classes of graphs with bounded expansion.
In Section~\ref{sec:largescatteredsets}, we use these results to show that in every planar graph, we can remove a small fraction of vertices so that the resulting graph contains a large set of vertices that are pairwise far apart (Section~\ref{sec:largescatteredsets}).
In Section~\ref{sec:coloringsandindependentsets},  we apply coloring theory developed by Dvo\v{r}\'ak, Kr\'al' and Thomas to such a large set $S$
of far apart vertices, to obtain a $3$-coloring of the graph with further constraints in the neighborhoods of vertices of $S$
which guarantee the existence of a large independent set.
In Section~\ref{sec:proofs}, we combine the results and give proofs of our theorems.
Finally, Section~\ref{sec-algo} describes an algorithm to find an independent set of size at least $(n+k)/3$ when the algorithm of Theorem~\ref{thm:mainfpt-trianglefree}
answers that such a set exists.

\subsection{Related Work}
\label{sec:relatedwork}
From the combinatorial side, the study of lower bounds on the independence number in triangle-free graphs has a long history.
Every $n$-vertex triangle-free graph has an independent set of size $\Omega(\sqrt{n\log n})$, and this bound is tight~\cite{Kim1995}.
Staton~\cite{Sta79} proved that every subcubic $n$-vertex triangle-free graph $G$ satisfies $\alpha(G)\ge \frac{5n}{14}$.
Furthermore, Heckman and Thomas~\cite{HeTh06} showed that if $G$ is additionally planar, then $\alpha(G)\ge \frac{3n}{8}$.

From the algorithmic side, the studying the complexity of maximization problems parameterized above polynomial-time computable lower bounds is an active area of research.
Since the influental survey by Mahajan et al.~\cite{MahajanEtAl2009}, research in this area has led to development of many new algorithmic techniques for fixed-parameter algorithms: algebraic methods~\cite{AlonEtAl2011,CrowstonEtAl2011}, probabilistic methods~\cite{GutinEtAl2010,GutinEtAl2012}, combinatorial methods~\cite{CrowstonEtAl2012,MnichZenklusen2012}, and methods based on linear programming~\cite{CyganEtAl2013}.

\subsection{Preliminaries}
Throughout, we consider graphs that are finite, undirected and loopless, and do not have parallel edges unless explicitly stated otherwise.
For a graph $G$, let~$V(G)$ denote its vertex set and $E(G)$ its set of edges.
The degree of a vertex $v \in V(G)$ is the number $\textnormal{deg}_G(v)$ of edges that are incident to it.
A graph is \emph{planar} if it admits an embedding in the plane such that no two edges cross; a \emph{plane graph} is an embedding of a planar graph without any edge crossings.

\section{Treewidth of Planar Graphs}
\label{sec-tw}
A \emph{tree decomposition} of a graph $G$ is a pair $(T,{\cal B})$, where $\cal B$ is a set of subsets of $V(G)$ (called the \emph{bags} of the decomposition) with $V(G)=\bigcup_{B\in{\cal B}} B$, and $T$ is a tree with vertex set ${\cal B}$, such that for each edge $uv\in E(G)$, there exists $B\in{\cal B}$ containing both $u$ and $v$, and for every $v\in V(G)$, the set $\{B\in {\cal B}: v\in B\}$ induces a connected subtree of $T$.
The width of the decomposition is the size of its largest bag minus one.
The \emph{tree-width} of a graph~$G$, denoted by $\mathsf{tw}(G)$, is the minimum width of its tree decompositions.
Note that the tree-width of a subgraph of $G$ is at most as large as the tree-width of $G$.

Determining tree-width exactly is $\mathsf{NP}$-hard; however, Bodlaender et al.~\cite{BodlaenderEtAl2013} proved the following.
\begin{theorem}
\label{thm-approx}
  There exists an algorithm that for an input $n$-vertex graph $G$ and integer $p\ge 1$, in time $2^{O(p)}n$ either outputs that the treewidth of $G$ is larger than $p$, or gives a tree decomposition of $G$ of width at most $5p+4$.
\end{theorem}

Let $G$ be a graph.
Suppose that $G=G'_1\cup G'_2$ and that $G_1\supseteq G'_1$ and $G_2\supseteq G'_2$ are graphs such that $V(G'_1\cap G'_2)$ induces a clique both in $G_1$ and in $G_2$.
In this situation, we say that $G$ is a \emph{clique-sum} of $G_1$ and $G_2$.
That is, $G$ is obtained from $G_1$ and $G_2$ by identifying corresponding vertices of cliques of the same size, and possibly removing edges afterwards.
We need the following well-known fact.
\begin{proposition}
\label{prop-csum}
  If $G$ is a clique-sum of $G_1$ and $G_2$, then $\mathsf{tw}(G)\le \max(\mathsf{tw}(G_1),\mathsf{tw}(G_2))$.
\end{proposition}

Robertson, Seymour and Thomas~\cite{RobertsonEtAl1994} gave a bound on the tree-width of planar graphs.
\begin{proposition}[\cite{RobertsonEtAl1994}]
\label{prop-twplansze}
  A planar graph with $n$ vertices has tree-width at most $6\sqrt{n}+1$.
\end{proposition}

We need to extend this bound to a related class of graphs.
\begin{lemma}
\label{lemma-addcross}
  Let $G_0$ be a triangle-free plane graph with $n$ vertices.
  Let $G$ be the graph obtained from $G_0$ by adding the edges $uw$ and $vx$
for each $4$-face $uvwx$.
  Then $G$ has tree-width at most $41\sqrt{n}$.
\end{lemma}
\begin{proof}
  By Euler's formula, since $G_0$ is triangle-free, it has at most $n$ faces.
  Let $R$ be the set of $4$-faces of $G_0$.
  Let $G_1$ be the graph obtained from $G_0$ by adding a vertex $v_f$ adjacent to all vertices of $f$ for each $f\in R$.
  Note that $G_1$ is planar and $|V(G_1)|\le 2n$, and thus by Proposition~\ref{prop-twplansze}, $G_1$ has a tree decomposition $(T, {\mathcal B}_1)$
with bags of size at most $6\sqrt{2n}+2$.
  Let $\mathcal B$ be obtained from ${\mathcal B}_1$ by replacing all appearances of $v_f$ by $V(f)$.
  Observe that $(T,\mathcal B)$ is a tree decomposition of $G$ with bags of size at most $24\sqrt{2n}+8\le 41\sqrt{n}+1$.
  Therefore, $G$ has tree-width at most $41\sqrt{n}$.
\end{proof}

Let~$G$ be a plane graph.
A subgraph $G_0$ of $G$ is \emph{$4$-swept} if $G_0$ has no separating $4$-cycles and every face of $G_0$ which is not a face of $G$ has length~$4$.
Let $s(G)$ denote the maximum number of vertices of a $4$-swept subgraph of $G$.

\begin{lemma}
\label{lemma:tws}
  Every plane graph $G$ has tree-width at most $41\sqrt{s(G)}$.
\end{lemma}
\begin{proof}
  Let $G_0$, \ldots, $G_{m-1}$ and $H_1$, \ldots, $H_m$ be subgraphs of $G$ obtained as follows.
  We set $G_0=G$.
  Suppose that $G_{i-1}$ was already constructed, for some $i\ge 1$.
  If $G_{i-1}$ contains no separating $4$-cycle, then let $m=i$ and $H_m=G_{i-1}$.
  Otherwise, let $C_i$ be a separating $4$-cycle in $G_{i-1}$ bounding as small open disk $\Lambda_i$ as possible.
  Let $H_i$ be the subgraph of $G_{i-1}$ drawn in the closure of $\Lambda_i$, and let $G_i$ be the subgraph of $G_{i-1}$ drawn in the complement of $\Lambda_i$.

  Observe that $H_1$, \ldots, $H_m$ are $4$-swept subgraphs of $G$, that $G=H_1\cup \ldots\cup H_m$, and that for $1\le i<j\le m$, the intersection of $H_i$ and $H_j$ is a part of a boundary of a $4$-face both in $H_i$ and $H_j$.
  For $1\le i\le m$, let $H'_i$ be the graph obtained from $H_i$ by adding the edges $uw$ and $vx$ for each $4$-face $uvwx$.
  Then $G$ is a clique-sum of $H'_1$, \ldots, $H'_m$, and $\mathsf{tw}(H_i)\le 41\sqrt{|V(H_i)|}\le 41\sqrt{s(G)}$ for $1\le i\le m$.    
  By Proposition~\ref{prop-csum}, we conclude that $G$ has tree-width at most $41\sqrt{s(G)}$.
\end{proof}

\section{Classes of Graphs with Bounded Expansion}
\label{sec:classesofgraphswithboundedexpansion}
In this section, we survey results on classes of graphs with bounded expansion that we need in the paper.
Let $G$ be a graph and let $r\ge0$ be an integer.
Let us recall that a graph $H$ is an \emph{$r$-shallow minor} of $G$ if $H$ can be obtained from a subgraph of $G$ by contracting vertex-disjoint subgraphs of radii at most $r$ and deleting the resulting loops and parallel edges.
Following Ne\v{s}et\v{r}il and Ossona de Mendez~\cite{NesetrilOssonadeMendez2008}, we denote by $\nabla_r(G)$ the maximum of $|E(G')|/|V(G')|$ over all $r$-shallow minors $G'$ of $G$.
Thus, $\nabla_0(G)$ is the maximum of $|E(G')|/|V(G')|$ taken over all subgraphs $G'$ of $G$.
Since every subgraph of a graph $G$ has a vertex of degree at most $2\nabla_0(G)$, we see that $G$ has an (acyclic) orientation with maximum in-degree at most~$2\nabla_0(G)$.

A class $\GG$ of graphs has \emph{bounded expansion} if there exist constants $c_0$, $c_1$, \ldots such that $\nabla_r(G)\le c_r$ for every $G\in\GG$ and $r\ge 0$.
Many natural classes of sparse graphs have bounded expansion.
Here, we only need that the class of planar graphs has bounded expansion; we refer to the book by Ne\v{s}et\v{r}il and Ossona de Mendez~\cite{nesbook} for
more information on the topic.

Let $D$ be a directed graph, and let $D'$ be a directed graph obtained from $D$ by adding, for every pair of vertices $x,y\in V(D)$,
\begin{itemize}
  \item the edge $xy$ if $D$ has no edge from $x$ to $y$ and there exists a vertex $z\in V(D)$ such that $D$ has an edge oriented from $x$ to $z$ and an edge oriented from $z$ to $y$ ({\em transitivity}), and
  \item either the edge $xy$ or the edge $yx$  if $x$ is not adjacent to $y$ and there exists a vertex $z$ such that $D$ has an edge oriented from $x$ to $z$ and an edge oriented from $y$ to $z$ ({\em fraternality}).
\end{itemize}
We call $D'$ an {\em oriented augmentation} of $D$.

Let $G$ be a graph. 
We construct a sequence $D_0$, \ldots, $D_\ell$ of directed graphs as follows.
Let $D_0$ be an orientation of~$G$ with maximum in-degree at most~$2\nabla_0(G)$.
For $1\le i\le \ell$, let $D_i$ be an oriented augmentation of $D_{i-1}$, in that the orientations of the edges added according to the fraternality rule are chosen so that the subgraph of $D_i$ formed by these edges has maximum in-degree at most
$2\nabla_0(G_i)$, where $G_i$ is the underlying undirected graph of $D_i$.
This is possible, because~$G_i$ itself has an orientation with maximum in-degree at most $2\nabla_0(G_i)$.
We say that~$D_{\ell}$ is an {\em $\ell$-th oriented augmentation} of $G$.
We use the following important property of classes of graphs with bounded expansion~\cite{NesetrilOssonadeMendez2008}.

\begin{proposition}
\label{prop:augmexp}
  Let $\GG$ be a class of graphs and let $\ell\ge 0$ be an integer.
  If $\GG$ has bounded expansion, then there exists an integer $m_d\ge 0$ such that all $\ell$-th oriented augmentations of graphs from $\GG$
  have maximum indegree at most $m_d$.
\end{proposition}

\section{Large Scattered Sets}
\label{sec:largescatteredsets}
For an integer $d\ge 1$ and a graph $G$, a vertex set $Q\subseteq V(G)$ is \emph{$d$-scattered} if the distance between any two vertices of $Q$ in $G$ is greater than~$d$.
In this section we discuss graph classes whose members admit large scattered subsets, possibly after removal of a small number of vertices.

For integers $d,m,N\ge 1$ and $r\ge 0$, a graph $G$ is \emph{$(d,r,m,N)$-wide} if for any set $S\subseteq V(G)$ of size at least~$N$, there exists a set $Q\subseteq S$ of size at least $m$ and a set $X\subseteq V(G)$ of size at most $r$ such that $Q$ is $d$-scattered in $G-X$. 
For integers $d\ge 1$ and $r\ge 0$, a class of graphs $\GG$ is \emph{$(d,r)$-wide} if for every $m$, there exists $N$ such that
every graph in $\GG$ is $(d,r,m,N)$-wide.
A class of graphs $\GG$ is \emph{uniformly quasi-wide} if for every integer $d\ge 1$ there exists an integer $r\ge 0$ such that $\GG$ is $(d,r)$-wide.
By results of Ne\v{s}et\v{r}il and Ossona de Mendez~\cite{npom-nd1}, every class of graphs with bounded expansion is uniformly quasi-wide (in fact, they prove a stronger result concerning nowhere-dense graph classes).

We need a variant of wideness where the value of $N$ is linear in $m$.
Note that there exist classes of graphs with bounded expansion which do not have this property---for instance, consider the class containing the graphs $H_a$ for each integer $a>0$, where $H_a$ is the disjoint union of $a$ stars with $a$ rays.
With $S = V(H_a)$, the largest possible $2$-scattered set after removal of $r$ vertices has size $ra+(a-r)=\Theta(\sqrt{|S|})$.
However, note that we can obtain a $d$-scattered set of size $a^2=|S|-a$ at the cost of removing $a$ vertices, which suggests the variant of wideness where the set of removed vertices is small relatively to the size of the $d$-scattered set.

For integers $d,t,K\ge 1$, a graph $G$ is \emph{$(d,t,K)$-fat} if for any set $S\subseteq V(G)$ there exists a set $Q\subseteq S$ of size at least $|S|/K$ and a set $X\subseteq V(G)\setminus Q$ of size at most $|Q|/t$ such that $Q$ is $d$-scattered in $G-X$.
A class of graphs $\GG$ is \emph{fat} if for all integers $d,t\ge 1$, there exists $K$ such that every graph in $\GG$ is $(d,t,K)$-fat.

We are going to show that every class of graphs with bounded expansion is fat.
However, first we need to derive the following auxiliary result.

\begin{lemma}
\label{thm:spec}
  Let $c,t\ge 1$ be integers, and let $K_0=c^{2c}(2c+2)^{c+1}t^c$.
  Let $G$ be a bipartite graph with parts $S$ and $Z$.
  If all vertices in $S$ have degree at most~$c$, then there exist $Q\subseteq S$ and $X\subseteq Z$ such that $|Q|\ge |S|/K_0$, $|X|\le |Q|/t$ and for all distinct $u,v\in Q$, all common neighbors of $u$ and $v$ belong to~$X$.
\end{lemma}
\begin{proof}
  For $0\le i\le c$, let $d_i=\dfrac{K_0}{c^{2i+1}(2c+2)^{i+1}t^i}$.
  Let $Z_0$ consist of vertices in~$Z$ of degree at least $d_0$ and for $1\le i\le c$, let $Z_i$ consist of the vertices in $Z$ of degree at least $d_i$, but less than~$d_{i-1}$.
  Since all vertices of $S$ have degree at most~$c$, for each $v\in S$ there exists $i\in\{0,\ldots, c\}$ such that $v$ has no neighbors in~$Z_i$.
  By the pigeonhole principle, there exists $i\in\{0,\ldots, c\}$ and a set $B\subseteq S$ of size at least $|S|/(c+1)$ such that no vertex of $B$ has a neigbor in $Z_i$.
  We set $X=Z_0\cup \ldots\cup Z_{i-1}$ and we choose $Q$ as an inclusion-wise maximal subset of $B$ such that no two vertices of $Q$ have a common neighbor in~$Z\setminus X$.

  First, let us estimate the size of $Q$.
  For each vertex $v\in B$, its neighbors either belong to~$X$ or have degree less than $d_i$, and thus at most $cd_i$ vertices are at distance exactly two from~$v$ in $G-X$.
  Since $Q$ is maximal, each vertex of~$B\setminus Q$ is at distance two from a vertex of $Q$ in $G-X$, and thus
  \begin{equation*}
    |Q| \ge \frac{|S|/(c+1)}{1+cd_i}
        \ge \frac{|S|/(c+1)}{2cd_i}
        \ge \frac{|S|}{2c(c+1)d_0}
          = \frac{|S|}{K_0} \enspace .
  \end{equation*}
  Note that we use the fact that $cd_i\ge cd_c=1$.

  If $i=0$, then $X=\emptyset$, and thus $Q$ and $X$ satisfy the conclusions of the lemma.
  If $i\ge 1$, then we need to estimate the size of $X$.
  Note that each vertex of $X$ has degree at least $d_{i-1}$, and thus $|X|d_{i-1}\le |E(G)|\le c|S|$.
  Consequently,
  \begin{equation*}
    \frac{|Q|}{|X|}\ge \frac{\frac{|S|}{2c(c+1)d_i}}{c|S|/d_{i-1}}=\frac{d_{i-1}}{2c^2(c+1)d_i}=t\enspace .\qquad \qedhere
  \end{equation*}
\end{proof}

For a path $P$, let $\ell(P)$ denote its length (the number of its edges).
We say that a path $P$ with directed edges is \emph{reduced} if either $\ell(P)=1$, or $\ell(P)=2$ and both of its edges are directed
away from the middle vertex of $P$.

\begin{lemma}
\label{thm:fat}
  Every class of graphs with bounded expansion is fat.
\end{lemma}
\begin{proof}
  Let $\GG$ be a class of graphs with bounded expansion and consider fixed integers $d,t\ge 1$.
  Let $\GG_d$ be the class of $d$-th oriented augmentations of graphs in~$\GG$, and let $m_d$ be the smallest integer such that the maximum in-degree of every graph in $\GG_d$ is at most $m_d$, which exists by Proposition~\ref{prop:augmexp}.
  Let $K_0=m_d^{2m_d}(2m_d+2)^{m_d+1}t^m_d$ and let $K=(2m_d+1)K_0$.
  We will show that every $G\in \GG$ is $(d,t,K)$-fat.

  Consider a set $S\subseteq V(G)$.
  Let $G_d$ be a $d$-th oriented augmentation of $G$ and let~$G_d'$ be the underlying undirected graph of $G_d$.
  Since $G_d$ has maximum in-degree at most $m_d$, it follows that the maximum average degree of $G_d'$ is at most~$2m_d$, and thus $G_d'$ has a proper coloring by at most $2m_d+1$ colors.
  By considering the intersections of the color classes of this coloring with $S$, we conclude that there exists a set $S_0\subseteq S$ of size at least $|S|/(2m_d+1)$ which is independent in~$G_d'$.
  Let $Z$ be the set of in-neighbors of vertices of~$S_0$ in~$G_d$.
  Let $H$ be the bipartite graph with parts~$S_0$ and $Z$ such that $sz$ is an edge of~$H$ if and only if $s\in S_0$, $z\in Z$ and $G_d$ contains an edge directed from $z$ to~$s$.
  Let $Q$ and $X$ be the sets obtained by applying Lemma~\ref{thm:spec} to $H$.
  Note that $|Q|\ge |S_0|/K_0\ge |S|/K$ and $|X|\le |Q|/t$ as required.

  It remains to show that $Q$ is $d$-scattered in $G-X$.
  Suppose that there exists a path $P_0\subseteq G-X$ of length at most $d$ between two vertices $u,v\in Q$.
  For $1\le i\le d-1$, let $G_i$ denote the intermediate $i$-th oriented augmentation of $G$ obtained during the construction of $G_d$.
  For $1\le i\le d$, let $P_i$ be a path between~$u$ and $v$ in the underlying undirected graph of $G_i$ such that $V(P_i)\subseteq V(P_0)$ and~$P_i$ is as short as possible.
  Note that $\ell(P_i)\le \ell(P_{i-1})$, and if $\ell(P_i)=\ell(P_{i-1})$, then $P_i$ (taken with the orientation of its edges as in $G_i$) is reduced.
  Since the length of $P_0$ is at most $d$, we conclude that $P_d$ with the orientation as in $G_d$ is reduced.
  Since~$Q$ is an independent set in $G_d'$, it follows that $\ell(P_d)\neq 1$.
  Therefore, $\ell(P_d)=2$ and the middle vertex $x$ of $P_d$ is a common in-neighbor of $u$ and $v$ in $G_d$.
  However, this implies that $x$ belongs to $X$, contrary to the assumption that $P$ is disjoint with $X$.
\end{proof}

\section{Colorings and Independent Sets}
\label{sec:coloringsandindependentsets}
Let us now turn our attention back to independent sets.
As we mentioned before, if $G$ is a $3$-colorable graph on $n$ vertices, then $G$ has an independent set of size at least $n/3$.
This bound can be improved when some vertices in the coloring have monochromatic neighborhood, since such vertices can be moved to two different color classes.

\begin{lemma}
\label{lemma:large}
  Let $G$ be a graph on $n$ vertices and let $Q,X\subseteq V(G)$ be disjoint sets.
  If $Q$ is an independent set of $G$ and $G - X$ has a $3$-coloring $\varphi$ such that the neighborhood of each vertex in $Q$ is monochromatic, then $\alpha(G)\ge \frac{n-|X|+|Q|}{3}$.
\end{lemma}
\begin{proof}
  For $i\in\{1,2,3\}$, let $Q_i$ be the set of vertices of $Q$ such that $\varphi$ assigns the color $i$ to their neighbors, and let $Q_0$ be the set of vertices of $Q$ that are isolated in $G-X$.
  Hence, $Q$ is the disjoint union of $Q_0$, $Q_1$, $Q_2$ and $Q_3$.
  Let $V_i$ be the set of vertices of $G-X-Q$ to which~$\varphi$ assigns the color $i$.
  Let $S_i=V_i\cup Q_0\cup Q_{i+1}\cup Q_{i+2}$, where $Q_4=Q_1$ and $Q_5=Q_2$, and note that~$S_i$ is an independent set in $G$.
  We have $|S_1|+|S_2|+|S_3| = (|V_1|+|V_2|+|V_3|)+2(|Q_1|+|Q_2|+|Q_3|)+3|Q_0|=(n-|X|-|Q|)+2|Q|+|Q_0|\ge n-|X|+|Q|$.
  Therefore, at least one of $S_1$, $S_2$ and $S_3$ has size at least $\frac{n-|X|+|Q|}{3}$.
\end{proof}

For a plane graph $G$, let $F(G)$ denote the set of faces of~$G$.
Consider a cycle $C\subset G$.
Removing~$C$ splits the plane into two open connected subsets, the bounded one is called the \emph{open interior} of $C$. 
The \emph{closed interior} of $C$ is the closure of the open interior of $C$.
The cycle $C$ is \emph{separating} if both the open interior of $C$ and the complement of the closed interior of $C$ contain a vertex of~$G$.
The following result of Dvo\v{r}\'ak, Kr\'al' and Thomas~\cite{DvorakKralThomas2013} is our main tool for obtaining colorings with monochromatic neighborhoods.

\begin{proposition}[\cite{DvorakKralThomas2013}]
\label{thm:3coloringextension}
  There exists an integer $D_0\ge 0$ such that for any triangle-free plane graph~$G$ without separating $4$-cycles, for any sets $Q_1\subseteq V(G)$ of vertices of degree at most $4$ and $Q_2\subseteq F(G)$ of $4$-faces such that the elements of $Q_1\cup Q_2$ have pairwise distance at least $D_0$, and for any $3$-coloring $\psi$ of the boundaries of the faces in $Q_2$, there exists a 3-coloring $\varphi$ of $G$ such that the neighborhood of each vertex in~$Q_1$ is monochromatic and the pattern of the coloring on each face in $Q_2$ is the same as in the coloring~$\psi$.
\end{proposition}
In the statement, two colorings have the same \emph{pattern} on a subgraph $F$ if they differ on $F$ only by a permutation of colors.

We use the following result by Gimbel and Thomassen~\cite{GimbelThomassen1997}.
\begin{proposition}[\cite{GimbelThomassen1997}]
\label{thm:gimbel}
  Let $G$ be a triangle-free planar graph and let $C=v_1v_2\ldots $ be an induced cycle of length at most $6$ in $G$.
  If a $3$-coloring $\psi$ of $C$ does not extend to a $3$-coloring of $G$, then $|C|=6$ and $\psi(v_1)=\psi(v_4)\neq \psi(v_2)=\psi(v_5)\neq \psi(v_3)=\psi(v_6)\neq\psi(v_1)$.
\end{proposition}

\begin{corollary}
\label{thm:mono3}
  Let $G$ be a triangle-free planar graph.
  For any vertex $v\in V(G)$ of degree at most~$3$, there exists a $3$-coloring of $G$ such that the neighborhood of~$v$ is monochromatic.
\end{corollary}
\begin{proof}
  Let $v_1$, \ldots, $v_t$ be the neighbors of $v$ in the cyclic order around $v$.
  The claim holds by Gr\"otzsch' theorem if $t\le 1$, and thus assume that $2\le t\le 3$.
  Let $G'$ be the graph obtained from $G$ by removing~$v$ and adding new vertices $u_1$, \ldots, $u_t$ and the edges of the cycle $C=v_1u_1v_2u_2\ldots v_tu_tv_1$.
  Note that $C$ is an induced cycle, since $G$ is triangle-free.
  Let $\psi(v_1)=\ldots=\psi(v_t)=1$ and $\psi(u_1)=\ldots=\psi(u_t)=2$.
  By Proposition~\ref{thm:gimbel}, $\psi$ extends to a $3$-coloring $\varphi$ of $G'$.
  By setting $\varphi(v)=2$, we extend $\varphi$ to a $3$-coloring of $G$ such that the neighborhood of $v$ is monochromatic.
\end{proof}

We need a variation of Proposition~\ref{thm:3coloringextension} which allows some separating $4$-cycles.
Let us recall that a subgraph $G_0$ of a plane graph $G$ is \emph{$4$-swept} if $G_0$ has no separating $4$-cycles and every face of $G_0$ which is not a face of $G$ has length~$4$.

\begin{lemma}
\label{thm:mono}
  There exists an integer $D_1\ge 1$ such that the following holds for any triangle-free plane graph~$G$ and any $4$-swept subgraph $G_0$ of $G$.
  Let $X,Q\subseteq V(G_0)$ be disjoint sets such that each vertex of~$Q$ has degree at most $4$ in $G_0$.
  If $Q$ is $D_1$-scattered in $G_0-X$, then~$G-X$ has a $3$-coloring such that at least $|Q|-6|X|$ vertices have monochromatic neighborhood and form an independent set.
\end{lemma}
\begin{proof}
  Let $D_1=D_0+4$, where $D_0$ is the constant of Proposition~\ref{thm:3coloringextension}.

  Let $R$ be the set of $4$-faces of $G_0$ which are not faces of $G$.
  Let $Z$ be the set of vertices $z\in Q$ such that there exists a face $x_1zx_2v\in R$ such that $x_1$ and $x_2$ belong to $X$.
  Let~$H$ be the graph (possibly with parallel edges) with vertex set $X$ such that two vertices $x_1$ and $x_2$ are adjacent if there exists a face $x_1v_1x_2v_2\in R$, for some $v_1,v_2\in V(G_0)$.
  Since~$G_0$ has no separating $4$-cycles, we conclude that either $G_0$ is isomorphic to~$K_{2,3}$ and $|X|\ge 2$, or $H$ has no parallel edges.  Since $H$ is a planar graph, it follows that $|E(H)|\le 3|X|$.
  Note that $|Z|\le 2|E(H)|\le 6|X|$.

  Let $Q_0=Q\setminus Z$.
  Let $Q_1$ be the set of vertices of $Q_0$ that are not incident with the faces of $R$.
  For each vertex $v\in Q_0\setminus Q_1$, we choose one incident $4$-face in $R$; let $Q'_2$ be the set of these faces.
  By the choice of $Z$, each face in $Q'_2$ is incident with exactly one vertex of $Q_0\setminus Q_1$,
  and thus $|Q'_2|=|Q_0\setminus Q_1|$.
  For each $f\in Q'_2$, let $G_f$ be the subgraph of $G$ drawn in the closure of $f$.
  By Euler's formula, there exists a vertex $v_f\in V(G_f)\setminus V(f)$ whose degree in $G$ is at most 3.
  Let $\psi_f$ be a $3$-coloring of $G_f$ such that the neighborhood of $v_f$ is monochromatic, which exists by Corollary~\ref{thm:mono3}.
  Let $I=\{v_f:f\in Q'_2\}$.

  Let $G_1$ be the graph obtained from $G_0-X$ as follows.
  For each face $f\in Q'_2$ whose boundary intersects~$X$, note that by the choice of $Z$, there exists a subpath~$P$ of the boundary walk of $f$ such that $X\cap V(F)$ are exactly the internal vertices of $P$.
  Add to $G_1$ a path of length $|P|$, with the same endvertices as $P$ and with new internal vertices of degree two.

  Note that each face in $Q'_2$ corresponds to a $4$-face of $G_1$; let $Q_2$ be the set of such faces of $G_1$.
  Observe that the distance in $G_1$ between any two elements of $Q_1\cup Q_2$ is at least $D_0$.
  Furthermore, for each $f\in Q'_2$, we can naturally interpret~$\psi_f$ as a coloring of the corresponding face of $Q_2$.

  By Proposition~\ref{thm:3coloringextension}, there is a 3-coloring $\varphi_0$ of $G_1$ such that the neighborhood of every vertex of~$Q_1$ is monochromatic and the pattern of $\varphi_0$ on every face $f\in Q_2$ is the same as the pattern of~$\psi_f$.
  By permuting the colors in the colorings $\psi_f:f\in Q_2$, we can assume that their restrictions to the boundaries of the faces in $Q_2$ are equal to the restriction of $\varphi_0$.
  For each $4$-face $f\in R\setminus Q_2$, let $\psi_f$ be a $3$-coloring of the subgraph of $G$ drawn in the closure of $f$ matching $\varphi_0$ on $f$, which exists by Proposition~\ref{thm:gimbel}.

  Let $\varphi$ be the union of $\varphi_0$ and the colorings $\psi_f$ for all $4$-faces $f$ of $G_0$, restricted to $V(G)\setminus X$.
  Note that $\varphi$ is a $3$-coloring of $G-X$ such that all vertices in $Q_1\cup I$ have monochromatic neighborhood in $\varphi$.
  The claim of this lemma holds, since $|Q_1\cup I|=|Q_0|$.
\end{proof}

\section{Proofs}
\label{sec:proofs}
It is well known that planar triangle-free graphs have many vertices of degree at most $4$.
\begin{lemma}
\label{thm:deg4}
  Any $n$-vertex planar triangle-free graph $G$ has at least $n/5$ vertices of degree at most $4$.
\end{lemma}
\begin{proof}
  Let $n_4$ denote the number of vertices of $G$ of degree at most $4$. 
  By Euler's formula, planar triangle-free graphs have average degree at most $4$, and thus $5(n-n_4)\le 4n$.
  The claim of the lemma follows.
\end{proof}

Let us recall that $s(G)$ denotes the maximum number of vertices of a $4$-swept subgraph of a plane graph $G$.
We can now state the result from which Theorems~\ref{thm:indsize} and~\ref{thm:large-ind} will be derived.
\begin{theorem}
\label{thm:gen}
  There exists a constant $c > 0$ such that every plane triangle-free graph $G$ on $n$ vertices has an independent set of size at least $\frac{n+cs(G)}{3}$.
\end{theorem}
\begin{proof}
  Let $D_1$ be the distance from Lemma~\ref{thm:mono}.
  By Lemma~\ref{thm:fat}, there exists a constant $K$ such that every planar graph is $(D_1,14,K)$-fat. 
  Let $c=\frac{1}{10K}$.

  Let $G_0$ be a $4$-swept subgraph of $G$ such that $|V(G_0)|=s(G)$.
  Let $S$ be the set of vertices of~$G_0$ of degree at most $4$; by Lemma~\ref{thm:deg4}, we have $|S|\ge s(G)/5$.
  Since $G_0$ is $(D_1,14,K)$-fat, there exist sets $Q\subseteq S$ and $X\subseteq V(G_0)\setminus Q$ such that~$Q$ is $D_1$-scattered in $G_0-X$, $|Q|\ge |S|/K\ge \frac{s(G)}{5K}$, and $|X|\le |Q|/14$.

  By Lemma~\ref{thm:mono}, $G-X$ has a proper $3$-coloring and an independent set $Q'$ with $|Q'|\ge |Q|-6|X|$ such that the neighborhood of each vertex in $Q'$ is monochromatic.
  By Lemma~\ref{lemma:large}, we have, as required,
  \begin{equation*}
    \alpha(G)\ge\dfrac{n-|X|+|Q|-6|X|}{3}\ge \dfrac{n+|Q|/2}{3}\ge \dfrac{n+\frac{s(G)}{10K}}{3}\enspace .\qedhere
  \end{equation*}
\end{proof}

\begin{proof}[Proof of Theorem~\ref{thm:large-ind}]
  Since $G$ has no separating $4$-cycles, we have $s(G) = n$.
  Therefore, $\alpha(G)\ge \frac{n+cn}{3}=\frac{n}{3/(1+c)}$ by Theorem~\ref{thm:gen}, and we can set $\varepsilon=\frac{3c}{1+c}$.
\end{proof}

Theorem~\ref{thm:indsize} is a corollary of Theorem~\ref{thm:gen} and Lemma~\ref{lemma:tws}.

\section{Finding A Large Independent Set}
\label{sec-algo}
In this section we describe an algorithm to actually find an independent set of size at least $(n+k)/3$ in a given $n$-vertex triangle-free planar graph possessing such a set.

We will need to be able to locate separating $4$-cycles efficiently.
\begin{lemma}
\label{lemma-findsep}
  There exists a linear-time algorithm that reports a separating $4$-cycle in a given plane graph~$G$, or decides that there is no separating $4$-cycle.
\end{lemma}
\begin{proof}
  Note that given edges of a $4$-cycle $C$ in $G$, we can test whether $C$ bounds a face in constant time (given a standard representation of the plane embedding of $G$), and thus we can decide whether $C$ is separating.

  Firstly, greedily find an ordering $v_1$, \ldots, $v_n$ of the vertices of $G$ so that for $1\le i\le n$, the vertex $v_i$ is adjacent to at most $5$ vertices in the set $\{v_1, \ldots, v_{i-1}\}$.
  Next, enumerate all $4$-cycles $v_{i_1}v_{i_2}v_{i_3}v_{i_4}$ such that
$i_2<\max(i_1,i_3)<i_4$.
  This can be done in linear time, by first choosing $v_{i_4}$, then $v_{i_1}$ and $v_{i_3}$ among at most $5$ neighbors of $v_{i_4}$ that precede it in the ordering, and finally $v_{i_2}$ among at most $5$ neighbors of $v_{\max(i_1,i_3)}$ that precede it in the ordering.
  Test each of these cycles, and if any is separating, report it.

  Hence, we can now assume that if $G$ contains a separating $4$-cycle, then we can label its vertices as $v_{i_1}v_{i_2}v_{i_3}v_{i_4}$ so that
$i_4>i_2>\max(i_1,i_3\}$.
  Enumerate all triples $(a,b,c)$ such that $a<b<c$ and $v_av_c,v_bv_c\in E(G)$ (note that there are at most $\frac{5}{2}n$ such triples) and sort the list of triples (compared lexicographically) in $O(n)$ time using radix sort.
  Group the triples according to the first two entries.
  For each such group $(a,b,c_1)$, $(a,b,c_2)$, \ldots, $(a,b,c_t)$, test all $4$-cycles of form $ac_ibc_j$ with $1\le i\le \min(t,4)$, and if any is separating, report it.
  Note that if $t\ge 4$, one of the tested $4$-cycles necessarily is separating.
\end{proof}

Suppose that the algorithm of Theorem~\ref{thm:mainfpt-trianglefree} decided that a planar triangle-free graph $G$ with $n$ vertices contains an independent set of size at least $(n+k)/3$.
In case that $G$ has tree-width $O(\sqrt{k})$, the algorithm for finding independence number of a graph with bounded tree-width can also easily return one of the largest independent sets.

Suppose that $G$ has tree-width $\Omega(\sqrt{k})$, i.e., $s(G)=\Omega(k)$.
In this case, we know that $G$ contains an independent set of size at least $(n+k)/3$ by Theorem~\ref{thm:gen}.
Let us now go over the proof of Theorem~\ref{thm:gen} in order to obtain an algorithm to report this set.

Firstly, we need to find a $4$-swept subgraph $G_0$ of $G$ of size $\Omega(k)$.  This can be accomplished by applying the procedure from the first paragraph
of the proof of Lemma~\ref{lemma:tws} to obtain $4$-swept subgraphs $H_1$, \ldots, $H_m$ of $G$ and choosing $G_0$ as the largest of these graphs.
In order to obtain each of the $4$-swept subgraphs, we need to locate a separating $4$-cycle with minimal interior.
One way to do so is to find an arbitrary separating $4$-cycle using Lemma~\ref{lemma-findsep}, then find an arbitrary separating $4$-cycle in the graph drawn in its closed interior, etc., until a $4$-cycle with no separating $4$-cycles in its interior is found.
In total, we can find $G_0$ in time $O(n^3)$.

Let $D_1$ be the distance from Lemma~\ref{thm:mono} and let $K$ be the constant given by Lemma~\ref{thm:fat} such that every planar graph is $(D_1,14,K)$-fat.  Let $S_0$ be the set of vertices of $G_0$ of degree at most $4$.
Let $S$ be an arbitrarily chosen subset of $S_0$ of size $2Kk$.
We now proceed as in the proof of Lemmas~\ref{thm:fat} and \ref{thm:spec}
to obtain $Q\subseteq S$ and $X\subseteq V(G_0)\setminus Q$ such that~$Q$ is $D_1$-scattered in $G_0-X$, $|Q|\ge |S|/K=2k$, and $|X|\le |Q|/14$.
Note that a $D_1$-th oriented augmentation of $G_0$ can be obtained in linear time; we refer to~\cite{grad2} for details.
Observe that the rest of the steps described in the proofs of the lemmas can be performed in linear time as well.

Let $Q'$ be the set of size $|Q|-6|X|$ obtained as in Lemma~\ref{thm:mono}, such that $G-X$ has a proper $3$-coloring in that the neighborhood of each vertex in $Q'$ is monochromatic.
Note that $|Q'|\le |Q|\le |S|=O(k)$.
Since all vertices of $Q'$ have degree at most four in $G-X$, this condition amounts to fixing colors of $O(k)$ vertices of $G-X$.
Hence, we can obtain such a $3$-coloring of $G-X$ using an algorithm by Dvo\v{r}\'{a}k et al.~\cite{trfree7} in time $2^{O(k)}n^2$.
Finally, we report the largest of the three independent sets constructed as in Lemma~\ref{lemma:large}.

Therefore, we can modify the algorithm of Theorem~\ref{thm:mainfpt-trianglefree} to also report the large independent set if it exists, at the expense of increasing the its time complexity to $O(n^3)+2^{O(k)}n^2$.
Let us remark that this can be improved to $2^{O(k)}n^2$ by a more involved implementation of the first step, using a semidynamic data structure described by Dvo\v{r}\'{a}k et al.~\cite{dvorak2013testing} to repeatedly find separating $4$-cycles.

\section{Discussion}
\label{sec:discussion}
We gave a fixed-parameter algorithm for finding an independent set of size at least $n/3 + k$ in triangle-free planar graphs on $n$ vertices, for every integer $k\ge 0$.
Let us remark that the subexponential dependence on $k$ in the running time of our algorithm is optimal, under the Exponential Time Hypothesis (this follows from a reduction by Madhavan~\cite{Madhavan1984}).

Several intriguing questions remain.
Does the problem admit a polynomial kernel?
That is, can any triangle-free planar graph on $n$ vertices be efficiently (in polynomial time) compressed to an equivalent graph $G'$ on $k^{O(1)}$ vertices?
Also, while we can decide the existence of the independent set in linear time (in $n$), we can only find such an independent set in quadratic time.
Can this be improved?

Unfortunately, it is unlikely our techniques could be used for {\sc Planar Independent Set-ATLB} in general planar graphs.
The analogue of Proposition~\ref{thm:3coloringextension} is false for general planar graphs, and there exist $n$-vertex planar graphs with largest independent set of size $n/4$ and arbitrarily large tree-width.

\bibliographystyle{abbrv}
\bibliography{draftbib}
\end{document}